\documentclass[12pt,draftcls,onecolumn]{IEEEtran}

\IEEEoverridecommandlockouts

\usepackage{latexsym}
\usepackage{graphicx}
\usepackage{epsfig}
\usepackage{subfigure}
\usepackage{array}
\usepackage{amsmath}
\usepackage{amssymb}
\usepackage{color}
\usepackage{amsthm}
\usepackage{enumerate}
\theoremstyle{plain} 
\newtheorem{thm}{Theorem}
\newtheorem{cor}{Corollary}
\newtheorem{lemma}{Lemma}
\newtheorem{prop}{Proposition}
\theoremstyle{definition}

\newtheorem{remark}{Remark}

\newcommand{\bLam}{{\bf \Lambda}}
\newcommand{\bW}{{\bf W}}

\newcommand{\bR}{{\bf R}}
\newcommand{\bI}{{\bf I}}
\newcommand{\bU}{{\bf U}}

\newcommand{\bZ}{{\bf Z}}

\newcommand{\bP}{{\bf P}}
\newcommand{\bu}{{\bf u}}

\newcommand{\bx}{{\bf x}}

\newcommand{\bz}{{\bf z}}
\newcommand{\bo}{{\bf 0}}

\def\bal#1\eal{\begin{align}#1\end{align}}
\newcommand{\bp} {\begin{proof}}
\newcommand{\ep} {\end{proof}}

\newcommand{{\Rb}} {\right)}

\newcommand{{\Rf}} {\right\}}

\begin{document}

\title{Optimal Signaling for Secure Communications over Gaussian MIMO Wiretap Channels}

\author{Sergey Loyka, Charalambos D. Charalambous
%
\thanks{The material in this paper was presented in part at the 2012 IEEE International Symposium on Information Theory, Boston, USA \cite{Loyka} and the 2013 IEEE International Symposium on Information Theory, Istanbul, Turkey \cite{Loyka-13}.}

\thanks{S. Loyka is with the School of Electrical Engineering and Computer Science, University of Ottawa, Ontario,
Canada, K1N 6N5, e-mail: sergey.loyka@ieee.org.}

\thanks{C.D. Charalambous is with the ECE Department, University of Cyprus, 75 Kallipoleos Avenue, P.O. Box 20537, Nicosia, 1678, Cyprus, e-mail: chadcha@ucy.ac.cy}
}
%
\maketitle
%
\begin{abstract}
Optimal signalling over the Gaussian MIMO wire-tap channel is studied under the total transmit power constraint. A closed-form solution for an optimal transmit covariance matrix is obtained when the channel is strictly degraded. In combination with the rank-1 solution, this provides the complete characterization of the optimal covariance for the case of two transmit antennas. The cases of weak eavesdropper and high SNR are considered. It is shown that the optimal covariance does not converge to a scaled identity in the high-SNR regime. Necessary optimality conditions and a tight upper bound on the rank of an optimal covariance matrix are established for the general case, along with a lower bound to the secrecy capacity, which is tight in a number of scenarios.
\end{abstract}
%

\section{Introduction}
\label{sec:introduction}

Multiple-input multiple-output (MIMO) architecture has gained prominence in both academia and industry as a spectrally-efficient approach to wireless communications [1]. With wide deployment of wireless networks, security issues have recently gained additional importance, including information-theoretic approach at the physical layer \cite{Bloch-11}. The physical-layer security in MIMO systems has been recently under active investigation \cite{Li}-\cite{Oggier}. It was demonstrated that Gaussian signaling is optimal over the Gaussian MIMO wire-tap channels (MIMO-WTC) \cite{Liu}-\cite{Oggier} and the optimal transmit covariance has been found for MISO systems \cite{Li}, the 2-2-1 system \cite{Shafiee}, for the parallel channels (where independent signalling is optimal) \cite{Khisti-08}\cite{Li'10}, all under the total power constraint, and in the general MIMO case under the transmit covariance matrix constraint \cite{Bustin}. The high-SNR regime (SNR $\rightarrow \infty$) has been studied in \cite{Khisti-2}. The general case is still an open problem under the total power constraint, since the underlying optimization problem is not convex and explicit solutions are not known, except for some special cases. In fact, an optimal covariance is not known even when the channel is degraded (so that the respective optimization problem is convex), except for the special cases mentioned above.

The main contribution of this paper is a closed-form solution for the optimal covariance when the latter is of full rank under the total power constraint at finite SNR and the conditions for this to be the case in Theorem \ref{thm:Full-rank}. The optimal covariance is shown to have some properties similar to those of the conventional water-filling, but with a few remarkable differences. In particular, the optimal covariance does not converge to a scaled identity in the high-SNR case and thus isotropic signaling is sub-optimal in this regime. Theorem \ref{thm:Full-rank}, in combination with the rank-1 solution, provides the complete characterization of the optimal covariance for the case of two transmit antennas (for any channel, degraded or not). The cases of high-SNR and of weak eavesdropper are elaborated in Corollaries \ref{cor:weak MIMO-WTC} and \ref{cor:high SNR}. An optimal covariance matrix for the general case (degraded or not) is characterized in Proposition \ref{prop:general case}, which shows that there is hidden convexity in the respective optimization problem, even when the channel is not degraded.

Proposition \ref{prop:necessary cond} gives a necessary condition of optimality for the general case, which is a transmission of the positive directions of the difference channel where the main channel is stronger than the eavesdropper one. This strengthens the earlier result in \cite{Li-2} (transmission on non-negative rather than positive directions). While the proof in \cite{Li-2} is rather straightforward and is based on a singular transformation (multiplication by a matrix that is singular when the covariance matrix is rank-deficient) of the KKT conditions, significantly more effort and a new approach are required to establish the stronger result. It avoids using a singular transformation (since some information about active signalling sub-space is irreversibly lost in the process) but relies on a novel property of positive semi-definite matrices (Lemma \ref{lemma:ABC}) and their block-partitioned representation to establish a property of dual variables from which the desired result follow. This result also allows one to establish a tighter bound on the rank of an optimal covariance matrix (Corollary \ref{cor:r+ bound}) than those available in the literature for the general case.

A lower bound on the secrecy capacity in the general case is established in Proposition \ref{prop:LB}. While the original problem is non-convex so that all powerful tools of convex optimization \cite{Boyd} cannot be used, the lower bound is expressed via a convex problem and thus can be solved efficiently by a numerical algorithm. This bound is tight (achieved with equality) in a number of cases: when the SNR is low, or when the legitimate and eavesdropper channels have the same right singular vectors, or when the channel is degraded, thus providing an additional insight into optimal signalling.

An upper bound on the rank of an optimal covariance matrix is given in Corollary \ref{cor:r+ bound} for the general case: the rank is bounded by the dimensionality of a positive sub-space of the difference channel. This bound is stronger than those in \cite{Oggier} and \cite{Li-2} and can be further used to identify the cases for which an optimal covariance is of rank one (when the difference channel has just one strictly positive eigenvalue). Since the rank-1 structure of optimal covariance is known (unlike the sufficient and necessary conditions under which an optimal covariance is of rank-1, for which only limited knowledge is available), this extends the earlier results in \cite{Li}\cite{Shafiee}\cite{Li-2}\cite{Li-11} and provides not only the rank but also an optimal covariance itself in those cases.

\section{Gaussian MIMO Wire-Tap Channel Model}
Let us consider the standard Gaussian MIMO-WTC model,
\begin{equation}
\label{eq1}
{\rm {\bf y}}_1 ={\rm {\bf H}}_1 {\rm {\bf x}}+{\rm {\bf \xi }}_1 , \quad {\rm {\bf y}}_2 ={\rm {\bf H}}_2 {\rm {\bf x}}+{\rm {\bf \xi }}_2
\end{equation}
where ${\rm {\bf x}}=[x_1 ,x_2 ,...x_m ]^T\in \mathcal{C}^{m,1}$ is the transmitted complex-valued signal vector of dimension $m \times 1$, ``T'' denotes transposition, ${\rm {\bf y}}_{1(2)} \in \mathcal{C}^{n,1}$ are the received vectors at the receiver (eavesdropper), ${\rm {\bf \xi }}_{1(2)} $ is the circularly-symmetric additive white Gaussian noise at the receiver (eavesdropper) normalized to unit variance in each dimension, ${\rm {\bf H}}_{1(2)} \in \mathcal{C}^{n_{1(2)} ,m}$ is the $n_{1(2)} \times m$ matrix of the complex channel gains between each Tx and each receive (eavesdropper) antenna, $n_{1(2)} $ and $m$ are the numbers of Rx (eavesdropper) and Tx antennas respectively. The channels ${\rm {\bf H}}_{1(2)} $ are assumed to be quasistatic (i.e., constant for a sufficiently long period of time so that the standard random coding arguments can be invoked within each coherence block) and frequency-flat, with full channel state information (CSI) at the Rx and Tx ends.

For a given transmit covariance matrix ${\rm {\bf R}}=E\left\{ {{\rm {\bf xx}}^+} \right\}$, where $E\left\{ \cdot \right\}$ is statistical expectation, the maximum achievable secure rate between the Tx and Rx (so that the leakage rate between the Tx and eavesdropper converges to zero) is \cite{Bustin}-\cite{Oggier}
\begin{equation}
\label{eq2}
C({\rm {\bf R}})=\ln \frac{\left| {{\rm {\bf I}}+{\rm {\bf W}}_1 {\rm {\bf R}}} \right|}{\left| {{\rm {\bf I}}+{\rm {\bf W}}_2 {\rm {\bf R}}} \right|}=C_1 ({\rm {\bf R}})-C_2 ({\rm {\bf R}})
\end{equation}
where negative $C({\rm {\bf R}})$ is interpreted as zero rate, ${\rm {\bf W}}_i ={\rm {\bf H}}_i^+ {\rm {\bf H}}_i $, $\left( \right)^+$ means Hermitian conjugation, $C_i(\bR)=\ln|\bI+\bW_i\bR|$. The secrecy capacity subject to the total Tx power constraint is
\begin{equation}
\label{eq3}
C_s =\mathop {\max }\limits_{{\rm {\bf R}}\ge 0} C({\rm {\bf R}}) \mbox{\ s.t.} \ tr{\rm {\bf R}}\le P_T
\end{equation}
where $P_T $ is the total transmit power (also the SNR since the noise is normalized). It is well-known that the problem in \eqref{eq3} is not convex in general and explicit solutions for the optimal Tx covariance are not known except for some special cases (e.g. low-SNR, MISO or parallel channels). It was conjectured in \cite{Oggier} that an optimal transmission in \eqref{eq3} is on the directions where the main channel is stronger than the eavesdropper one (i.e. on the positive directions of the difference channel ${\rm {\bf W}}_1 -{\rm {\bf W}}_2$). A similar conclusion, albeit in a different (indirect) form, has been obtained in \cite{Khisti-2} using the degraded channel approach.

Theorem \ref{thm:Full-rank} below gives an explicit, closed-form solution for the optimal full-rank covariance in \eqref{eq3} at finite SNR. A number of additional insights and properties follow.


\section{Closed-Form Solutions}
\label{sec:Full-rank}

In this section, we consider the problem in \eqref{eq3} and obtain its closed-form solutions. The following theorem establishes the optimal covariance ${\rm {\bf R}}^\ast $ for the strictly-degraded channel,  ${\rm {\bf W}}_1 >{\rm {\bf W}}_2$, where ${\bf A} >  {\bf B}$ means that $ {\bf A} - {\bf B}$ is positive definite.

\vspace*{0.5\baselineskip}
\begin{thm}
\label{thm:Full-rank}
Let ${\rm {\bf W}}_1 >{\rm {\bf W}}_2$ and $P_T > P_{T0} $, where $P_{T0} $ is a threshold power given by \eqref{eq.T2.5}. Then, ${\rm {\bf R}}^\ast $ is of full rank and is given by:
\begin{equation}
\label{eq9}
{\rm {\bf R}}^\ast ={\rm {\bf U\Lambda }}_1 {\rm {\bf U}}^+-{\rm {\bf W}}_1^{-1}
\end{equation}
where the columns of the unitary matrix ${\rm {\bf U}}$ are the  eigenvectors of ${\bf Z} = {\bf W}_2 + {\bf W}_2 ({\bf W}_1 - {\bf W}_2)^{-1} {\bf W}_2$, ${\rm {\bf \Lambda }}_1 =diag\{\lambda _{1i} \}>{\rm {\bf 0}}$ is a diagonal positive-definite matrix,
\begin{equation}
\label{eq10}
\lambda _{1i} =\frac{2}{\lambda}\left( {\sqrt {1+\frac{4\mu _i }{\lambda }} +1} \right)^{-1}
\end{equation}
and $\mu _i \ge 0$ are the eigenvalues of ${\bf Z}$; $\lambda >0$ is found from the total power constraint $tr{\rm {\bf R}}^\ast =P_T $ as a unique solution of the following equation:
\begin{equation}
\label{eq11}
\frac{2}{\lambda} \sum\limits_i \left( {\sqrt {1+\frac{4\mu _i }{\lambda }} +1} \right)^{-1} =P_T + tr{\rm {\bf W}}_1^{-1}
\end{equation}
The corresponding secrecy capacity is
\begin{equation}
\label{eq12}
C_s = \ln \frac{\left| {\rm {\bf W}}_{1} \right| | {\rm {\bf \Lambda}}_1|}{| \bI - {\rm {\bf W}}_{2} ( \bW_1^{-1} - \bU\bLam_1\bU^+ )|} =
\ln \frac{\left| {\rm {\bf W}}_{1} \right|}{\left| {\rm {\bf W}}_{2} \right|} + \ln \frac{| {\rm {\bf \Lambda}}_1|}{| {\rm {\bf \Lambda}}_2 |}
\end{equation}
where ${\rm {\bf \Lambda}}_2 = {\rm {\bf \Lambda}}_1 + diag\{\mu_i^{-1}\}$ and 2nd equality holds when $\bW_2 > 0$. $P_{T0} $ can be expressed as follows:
\begin{equation}
\label{eq.T2.5}
P_{T0} = \frac{2(\mu_1+\lambda_{min})}{\lambda_{min}^2} \sum\limits_i \left( {\sqrt {1+\frac{4\mu _i(\mu_1+\lambda_{min}) }{\lambda_{min}^2 }} +1} \right)^{-1} - tr\bW_1^{-1}
\end{equation}
where $\lambda_{min}$ is the minimum eigenvalue of $\bW_1$ and $\mu_1$ is the maximum eigenvalue of $\bZ$.
\end{thm}
\begin{proof}
See Appendix.
\end{proof}
\vspace*{0.5\baselineskip}

It should be pointed out that Theorem \ref{thm:Full-rank} gives an exact (not approximate) optimal covariance at finite SNR ($P_T \rightarrow \infty$ is not required) since $P_{T0}$ is a finite constant that depends only on ${\rm {\bf W}}_1$ and ${\rm {\bf W}}_2$ and this constant is small in some cases:
it follows from \eqref{eq.T2.5} that $P_{T0}\rightarrow 0$ if $\lambda_{min}\rightarrow \infty$, i.e. $P_{T0}$ is small if $\lambda_{min}$ is large. In particular, $P_{T0}$ can be upper bounded as
\begin{equation}
P_{T0} \le \frac{m\mu_1}{\lambda_{min}^2} + \frac{m-1}{\lambda_{min}}
\end{equation}
and if $\lambda_{min} \gg \mu_1$, then
\begin{equation}
P_{T0} \approx \frac{m}{\lambda_{min}} - tr\bW_1^{-1} \le \frac{m-1}{\lambda_{min}} \le 1
\end{equation}
where the last inequality holds if $\lambda_{min} \ge m-1$. Fig. 1 illustrates this case. On the other hand, when $\bW_1-\bW_2$ approaches a singular matrix, it follows that $P_{T0}\rightarrow\infty$, so that $P_{T0}$ is large iff $\bW_1-\bW_2$ is close to singular.

Theorem \ref{thm:Full-rank}, in combination with rank-1 solution in \eqref{eq 4a}, provides the complete solution for the optimal covariance in the $m=2$ case: if the channel is not strictly degraded or if the SNR is not above the threshold, the rank-1 solution in \eqref{eq 4a} applies; otherwise, Theorem \ref{thm:Full-rank} applies. Fig. 1 illustrates this for the following channel:
\bal
\label{eq:example}
\bW_1=
\left[
\begin{array}{cc}
    1.5 & 0.5\\
    0.5 & 1.5\\
\end{array}
\right],\
\bW_2=
\left[
\begin{array}{cc}
    0.35 & 0.15\\
    0.15 & 0.35\\
\end{array}
\right]
\eal
Note that the transition to full-rank covariance takes place at low SNR of about -6 dB, i.e. $P_{T0}$ is not high at all in this case.

We further observe that 1st term in \eqref{eq12} $C_{\infty} =\ln \frac{\left| {\rm {\bf W}}_{1} \right|}{\left| {\rm {\bf W}}_{2} \right|}$ is SNR-independent and the 2nd one $\Delta C = \ln \frac{| {\rm {\bf \Lambda}}_1|}{| {\rm {\bf \Lambda}}_2 |} < 0$ monotonically increases with the SNR. Furthermore, $C_s \rightarrow C_{\infty}$, $\Delta C \rightarrow 0$ as $P_T \rightarrow \infty$, in agreement with Theorem 2 in \cite{Khisti-2}. This is also clear from Fig. 1.

\begin{figure}[htbp]
\centerline{\includegraphics[width=3in]{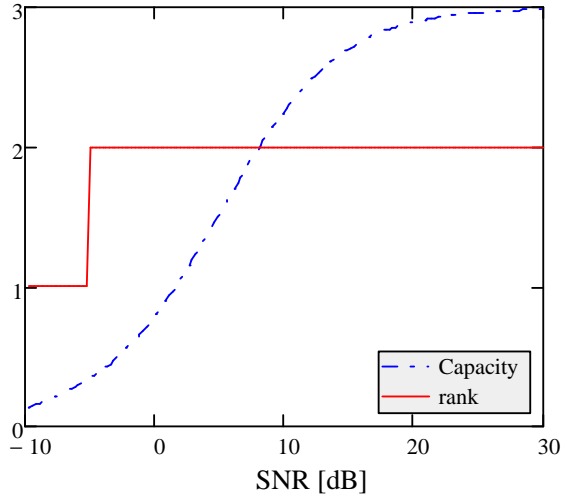}} \caption{Secrecy capacity and the rank of $\bR^*$ vs. SNR [dB] for the channel in \eqref{eq:example}. The transition to full-rank covariance takes place at about -6 dB.}
\label{fig:Fig_1}
\end{figure}

Note also that the second term in (\ref{eq9}) de-emphasizes weak eigenmodes of ${\rm {\bf W}}_1$. Since $\lambda $ is monotonically decreasing as $P_T $ increases (this follows from (\ref{eq11})), $\lambda _{1i} $ monotonically increases with $P_T $, and approaches $\lambda _{1i} \approx 1/\sqrt {\mu _i\lambda } $ at sufficiently high SNR, which is in contrast with the conventional water-filling (WF), where the uniform power allocation is optimal at high SNR. Furthermore, it follows from (\ref{eq10}) that $\lambda _{1i} $ decreases with $\mu _i $,
i.e. stronger eigenmodes of ${\rm {\bf W}}_2^{-1} -{\rm {\bf W}}_1^{-1}=\bZ^{-1} $ (which correspond to larger eigenmodes of ${\rm {\bf W}}_1 $ and weaker ones of ${\rm {\bf W}}_2 )$ receive larger power allocation, which follows the same tendency as the conventional WF. It further follows from \eqref{eq9} that when ${\bf{W}}_1$ and ${\bf{W}}_2$  have the same eigenvectors, ${\bf{R}}^*$ also has the same eigenvectors, i.e. the optimal signaling is on the eigenvectors of ${\bf{W}}_{1(2)}$. While the necessary condition for full-rank ${\rm {\bf R}}^\ast $ (${\rm {\bf W}}_1 >{\rm {\bf W}}_2  )$ has been obtained before in \cite{Oggier}, no solution was found for ${\rm {\bf R}}^\ast$, which is given in Theorem \ref{thm:Full-rank} here.

The case of singular ${\bf W}_1$ can also be included by observing that, under certain conditions, ${\bf R}^\ast$ puts no power on the null space of ${\bf W}_1$ so that all matrices can be projected, without loss of generality, on the positive eigenspace of ${\bf W}_1$ and Theorem \ref{thm:Full-rank} will apply. The following Proposition makes this precise.

\vspace*{0.5\baselineskip}
\begin{prop}
\label{Prop.3-3}
Consider the problem in \eqref{eq3} when $\mathcal{N}({\bf W}_1) \in \mathcal{N}({\bf W}_2)$, where $\mathcal{N}({\bf W}) = \{{\bf x}: {\bf W x = 0}\}$ is the null space of matrix ${\bf W}$ \cite{Zhang}, and assume that
\begin{align}
\label{eq.P3-3.1}
{\bf x}^+ ({\bf W}_1 - {\bf W}_2) {\bf x} > 0 \ \forall {\bf x} \in \mathcal{N}_{\bot},\ \bf{x} \neq 0,
\end{align}
where $\mathcal{N}_{\bot}$ is orthogonal complement of $\mathcal{N}({\bf W}_1)$, i.e. ${\bf W}_1 - {\bf W}_2$ is positive definite on $\mathcal{N}_{\bot}$. When the SNR exceeds a threshold (as in Theorem \ref{thm:Full-rank}), the optimal covariance in \eqref{eq3} is
\begin{align}
{\bf R}^\ast = {\bf U}_{\bot} {\bf \tilde{R}}^\ast {\bf U}_{\bot}^+
\end{align}
where ${\bf \tilde{R}}^\ast$ is the optimal covariance of Theorem \ref{thm:Full-rank} when applied to the projected matrices $\tilde{\bW}_i = {\bf U}_{\bot}^+ {\bf W}_i {\bf U}_{\bot}$ and the columns of semi-unitary matrix ${\bf U}_{\bot}$ form an orthonormal basis of $\mathcal{N}_{\bot}$. Furthermore, $rank({\bf R}^\ast) = rank(\bW_1)$.
\end{prop}
\begin{proof}
Observe that ${\bf W}_i {\bf x} = {\bf W}_i {\bf x}_{\bot}$, where ${\bf x}_{\bot} = {\bf U}_{\bot} {\bf U}_{\bot}^+{\bf x}$ is the orthogonal projection of $\bf x$ on $\mathcal{N}_{\bot}$, so that
\begin{align}
\notag
|{\bf I}+{\bf W}_i{\bf R}| &= |{\bf I}+{\bf W}_i{\bf U}_{\bot} {\bf U}_{\bot}^+{\bf R}{\bf U}_{\bot} {\bf U}_{\bot}^+| \\
&= |{\bf I}+{\bf U}_{\bot}^+{\bf W}_i{\bf U}_{\bot} {\bf U}_{\bot}^+{\bf R}{\bf U}_{\bot}|
\end{align}
and $tr({\bf U}_{\bot}^+{\bf R}{\bf U}_{\bot}) \le tr({\bf R})$ so that one can use the projected matrices ${\bf \tilde{R}} = {\bf U}_{\bot}^+{\bf R}{\bf U}_{\bot}, {\bf \tilde{W}}_i = {\bf U}_{\bot}^+{\bf W}_i{\bf U}_{\bot}$ in Theorem \ref{thm:Full-rank} to obtain the desired solution. \eqref{eq.P3-3.1} insures that ${\bf \tilde{W}}_1 - {\bf \tilde{W}}_2 > 0$ so that Theorem \ref{thm:Full-rank} applies.
\end{proof}

With Proposition \ref{Prop.3-3} in mind, the conditions of Theorem \ref{thm:Full-rank} are both sufficient and necessary (except for the power threshold $P_{T0}$ which may be less than in \eqref{eq.T2.5}) for an optimal covariance to be of full-rank.

It is instructive to consider the case when the required channel is much stronger than the eavesdropper one, ${\rm {\bf W}}_1 \gg {\rm {\bf W}}_2  $, meaning that all eigenvalues of ${\rm {\bf W}}_1 $ are much larger than those of ${\rm {\bf W}}_2$.

\begin{cor}
\label{cor:weak MIMO-WTC}
Consider the MIMO-WTC in (\ref{eq1}) under the conditions of Theorem \ref{thm:Full-rank} and when the eavesdropper channel is much weaker than the required one,
\begin{equation}
\label{eq13}
\lambda _i ({\rm {\bf W}}_2  )\ll m(P_T +tr{\rm {\bf W}}_1^{-1} )^{-1}/4
\end{equation}
where $\lambda_i ({\rm {\bf W}}_2)$ is $i$-th eigenvalue of ${\rm {\bf W}}_2$,  e.g. when ${\rm {\bf W}}_2  \to {\rm {\bf 0}}$ and fixed ${\rm {\bf W}}_1$. Then the optimal covariance in (\ref{eq9}) becomes
\begin{equation}
\label{eq14}
{\rm {\bf R}}^\ast \approx {\rm {\bf U}}_1 (\lambda ^{-1}{\rm {\bf I}}-{\rm {\bf D}}_1^{-1} ){\rm {\bf U}}_1^+ -\lambda ^{-2}{\rm {\bf W}}_2
\end{equation}
where ${\rm {\bf W}}_1 ={\rm {\bf U}}_1 {\rm {\bf D}}_1 {\rm {\bf U}}_1^+ $ is the eigenvalue decomposition, so that the columns of ${\rm {\bf U}}_1 $ are the eigenvectors, and the diagonal entries of ${\rm {\bf D}}_1 $ are the eigenvalues.
\end{cor}
\begin{proof} See Appendix.
\end{proof}

An interpretation of (\ref{eq14}) is immediate: the first term is the standard water-filling on the eigenmodes of ${\rm {\bf W}}_1 $ (which is the capacity-achieving strategy for the regular MIMO channel) and the second term is a correction due to the secrecy requirement: those modes that spill over into the eavesdropper channel get less power to accommodate the secrecy requirement.

Let us know consider the high-SNR regime.

\begin{cor}
\label{cor:high SNR}
When $\bW_2 > 0$, the optimal covariance ${\rm {\bf R}}^\ast $ in (\ref{eq9}) in the high-SNR regime
\begin{equation}
\label{eq15}
P_T \gg \mu_m^{-1/2} \sum\nolimits_i \mu_i^{-1/2}
\end{equation}
(e.g. when $P_T \to \infty )$, where $\mu_m =\min_i \mu_i$, simplifies to
\begin{equation}
\label{eq16}
{\rm {\bf R}}^{\ast} \approx {\rm {\bf U}} diag\{d_i \} {\rm {\bf U}}^+ , \quad d_i =\frac{P_T \mu_i^{-1/2} }{\sum\nolimits_i \mu_i^{-1/2}}
\end{equation}
The corresponding secrecy capacity is
\begin{equation}
\label{eq17}
C_s \approx \ln \frac{\left| {\rm {\bf W}}_{1} \right|}{\left| {\rm {\bf W}}_{2} \right|} - \frac{1}{P_T} \left( \sum\nolimits_i \frac{1}{\sqrt{\mu_i}}  \right)^2
\end{equation}
where we have neglected 2nd and higher order effects in $1/P_T$.
\end{cor}
\begin{proof}
Follows from Theorem \ref{thm:Full-rank} along the same lines as that of
Corollary \ref{cor:weak MIMO-WTC}.
\end{proof}

Note that the optimal signaling is on the eigenmodes of ${\rm {\bf
W}}_2^{-1} -{\rm {\bf W}}_1^{-1} $ with the optimal power allocation given by $\{d_i \}$. This somewhat resembles the conventional water-filling, but also has a remarkable difference: unlike the conventional WF, the secure WF in (\ref{eq16}) does not converge to the uniform one in the high-SNR regime\footnote{The sub-optimality of the isotropic signalling suggested in Theorem 2 of \cite{Khisti-2} is hiding in the $o(1)$ term there. 2nd term of Eq. \eqref{eq17} above refines that $o(1)$ term.}. However, strong eigenmodes of ${\rm {\bf W}}_2^{-1} -{\rm {\bf W}}_1^{-1} $ (which corresponds to weak modes of ${\rm {\bf W}}_2  $ and strong ones of ${\rm {\bf W}}_1 )$ do get more power, albeit in a form different from that
of the conventional WF.

While Theorem \ref{thm:Full-rank} gives a closed-form full-rank optimal covariance for the strictly degraded channel, the general case remains an open problem. The proposition below provides a characterization of an optimal covariance for the general case.

\begin{prop}
\label{prop:general case}
Consider the general Gaussian MIMO-WTC (not necessarily degraded). Let the columns of semi-unitary matrix $\bU_a$ span the same subspace as the columns of optimal covariance $\bR^*$ in \eqref{eq3}: $span\{\bU_a\}=span\{\bR^*\}$. Then, the optimal covariance can be expressed in the following form:
\bal
\bR^* = \bU_a \bR' \bU_a^+
\eal
where $\bR'$ is given by Theorem \ref{thm:Full-rank} with the substitutions $\bW_i \rightarrow \tilde{\bW}_i = \bU_a^+ \bW_i \bU_a$ (i.e. applied to the channels projected on $span\{\bU_a\}$), and $\tilde{\bW}_1 > \tilde{\bW}_2$.
\end{prop}
\begin{proof}
See Appendix.
\end{proof}

\begin{remark}
Proposition \ref{prop:general case} gives a closed-from solution for the general (non-degraded) case provided that the active subspace (i.e. the subspace spanned by the columns or active eigenvectors of $\bR^*$) is already known. Note that the knowledge of eigenvectors of $\bR^*$ is not required, but only the subspace they span. This in fact splits the entire problem $\mathcal{P}$ into two sub-problems $\mathcal{P}_1$ and $\mathcal{P}_2$:
\bal
\mathcal{P} = \mathcal{P}_1 \times \mathcal{P}_2
\eal
where $\mathcal{P}_1$ is a non-convex problem of finding the active sub-space (or the active eigenvectors) and $\mathcal{P}_2$ is the convex problem of finding the optimal covariance based on the found active subspace, hence revealing the hidden convexity in the original non-convex problem $\mathcal{P}$. While $\mathcal{P}_2$ is \textit{always} convex, $\mathcal{P}_1$ and thus $\mathcal{P}$ become convex when the channel is degraded.
\end{remark}

\section{Necessary Optimality Conditions and Properties}

In this section, we establish the necessary optimality conditions for the problem in \eqref{eq3} and, based on these conditions, some properties of the optimal solutions when the latter are rank-deficient. In particular, we establish an upper bound on the rank of optimal covariance matrix which is tighter than the known bounds. In some cases, this bound results in an explicit closed-form solution for the optimal covariance.

The following Proposition gives a necessary condition of the optimality in \eqref{eq3}.

\begin{prop}
\label{prop:necessary cond}
Let ${\rm {\bf R}}^{\ast} $ be an optimal covariance in \eqref{eq3} and let ${\bf U}_{r+}$ be a semi-unitary matrix whose columns are the active eigenvectors $\{{\bf u}_{i+}\}$ (i.e. corresponding to positive eigenvalues) of $\bR^*$. Then, the following holds:
\begin{equation}
\label{eq4a}
{\bf U}_{r+}^+ ({\bf W}_1 -{\bf W}_2 ){\bf U}_{r+} > \bf 0
\end{equation}
so that
\begin{equation}
\label{eq4}
\bx^+ (\bW_1 - \bW_2)\bx > \bo \ \forall \bx \in span \{\bu_{i+}\}
\end{equation}
i.e. a necessary condition for an optimal signaling strategy in \eqref{eq3} is to transit over the positive directions of ${\rm {\bf W}}_1 -{\rm {\bf W}}_2 $ (where the legitimate channel is stronger than the eavesdropper)\footnote{After the conference version of this paper has been submitted, we were informed that a weaker result ($\ge$ instead of $>$) was established in \cite{Li-2}. The proof in \cite{Li-2} is based on a singular transformation (multiplication by a singular matrix when $\bR$ is singular), so that some information about the active signalling sub-space is lost and strict inequality cannot be established. On the other hand, we avoid using such transformation and base our proof on some novel properties of positive semi-definite matrices (see Lemma 2) and their block-partitioned representation so that the active signaling sub-space can be characterized more precisely and a tighter upper bound on the rank of an optimal covariance can be established.}.

\end{prop}
\begin{proof}
See the Appendix.
\end{proof}

It was demonstrated in \cite{Oggier} that $rank({\rm {\bf R}}^\ast )<m$ unless ${\rm {\bf W}}_1 >{\rm {\bf W}}_2 $, i.e. an optimal transmission is of low-rank over a non-degraded channel. The Corollary below gives more precise characterization.

\begin{cor}
\label{cor:r+ bound}
Let ${\rm {\bf W}}_1 -{\rm {\bf W}}_2 ={\rm {\bf W}}_+ +{\rm {\bf W}}_- $, where ${\rm {\bf W}}_{+(-)} $ collects positive (negative and zero) eigenmodes of ${\rm {\bf W}}_1 -{\rm {\bf W}}_2 $ (found from its eigenvalue decomposition). Then,
\bal
rank({\rm {\bf R}}^\ast )\le rank({\rm {\bf W}}_+ )\le m,
\eal
i.e. the rank of an optimal covariance ${\rm {\bf R}}^\ast $ does not exceed the number of positive eigenvalues of ${\rm {\bf W}}_1 -{\rm {\bf W}}_2 $ (the rank of ${\rm {\bf W}}_+ )$.
\end{cor}
\begin{proof}
We need the following technical Lemma, which is a direct consequence of Corollary 4.5.11 in \cite{Horn}:

\begin{lemma}
\label{lem:r+}
Let $\bf A$ be Hermitian and $r_+({\bf A})$ be its number of positive eigenvalues. Then $ r_+({\bf S^+ A S}) \leq r_+({\bf A})$, where $\bf S$ is any matrix of appropriate size.
\end{lemma}

Lemma \ref{lem:r+} says that applying the transformation $\bf S^+ A S$ to $\bf A$ cannot increase the number of its positive eigenvalues (since $\bf S$ can be singular; this number stays the same if $\bf S$ is full rank). Using this Lemma with ${\bf S} = {\bf R}^{*1/2}$ and ${\bf A} = {\bf W}_1 - {\bf W}_2 + {\bf M}$, one obtains:
\begin{align}
r_+({\bf R}^*) &= r_+({\bf R}^{*1/2}({\bf W}_1 - {\bf W}_2 + {\bf M}){\bf R}^{*1/2}) \\
&= r_+({\bf R}^{*1/2}({\bf W}_1 - {\bf W}_2){\bf R}^{*1/2}) \\
&\leq r_+({\bf W}_1 - {\bf W}_2) = rank({\bf W}_+)
\end{align}
where 1st equality follows from the fact that ${\bf W}_1 - {\bf W}_2 + {\bf M} > \bf 0$ (which has been established in the proof of Proposition \ref{prop:necessary cond}), 2nd equality follows from ${\bf M}{\bf R}^*=0$, and the inequality follows from Lemma \ref{lem:r+}.
\end{proof}

Note that the rank bound in Corollary \ref{cor:r+ bound} is stronger than the corresponding bound in \cite{Li-2}, $rank(\bR^*) \le rank(\bW_1-\bW_2)$, especially when the difference matrix $\bW_1-\bW_2$ has many negative eigenvalues (e.g. when the eigenvalues of $\bW_1-\bW_2$ are $\{1,-1,..,-1\}$, the bound in \cite{Li-2} is trivial: $rank(\bR^*) \le m$, while our bound gives the true rank: $rank(\bR^*)=1$).

When $rank({\bf W}_+ ) =1$, the optimal covariance ${\bf R}^\ast$ is of rank-1 from Corollary \ref{cor:r+ bound} and hence the capacity and the covariance follow from \eqref{eq3}\footnote{This result has been obtained before, albeit in a different way, in \cite{Li-2}. Note however, that our result here is stronger: it does not require $\bW_1-\bW_2$ to be non-singular while \cite{Li-2} does, so that the latter result does not apply when the eigenvalues of $\bW_1-\bW_2$ are e.g. $\{1,0,..,0,-1,..,-1\}$ while our result does apply to such scenario.}:
\begin{equation}
\label{eq 4a}
C_s = \ln \lambda_1, \ {\bf R}^\ast = P_T {\bf u}_1 {\bf u}_1^+
\end{equation}
where $\lambda_1$, ${\bf u}_1$ are the largest eigenvalue and corresponding eigenvector of $({\bf I}+P_T{\bf W}_2)^{-1}({\bf I}+P_T{\bf W}_1)$ or, equivalently, the largest generalized eigenvalue and corresponding eigenvector of $({\bf I}+P_T{\bf W}_1, {\bf I}+P_T{\bf W}_2)$, so that transmit beamforming on ${\bf u}_1$ is the optimal strategy. Note that this result is more general than those in \cite{Li}\cite{Shafiee} as the latter two apply to a single antenna channel (either at the receiver or eavesdropper) while the result above holds for any number of antennas at any end. Furthermore, the signaling in \eqref{eq 4a} is also optimal for any $rank({\bf W}_+ )\ge 1$ at sufficiently small SNR, where $\lambda_1$, ${\bf u}_1$ become the largest eigenvalue and corresponding eigenvector of the difference channel ${\bf W}_1 - {\bf W}_2$.

The following Proposition establishes a lower bound to the non-convex problem in \eqref{eq3} via a convex optimization problem (for any channel, degraded or not).

\begin{prop}
\label{prop:LB}
The secrecy capacity can be lower bounded as follows:
\bal
\label{eq8a}
C_s \ge  \max_{\bR \ge 0} C_+(\bR)\ \mathrm{s.t.} \ tr\bR \le P_T,
\eal
where
\begin{equation}
\label{eq6}
C_+ ({\rm {\bf R}})=\ln \frac{\left| {{\rm {\bf I}}+{\rm {\bf W}}_{1+} {\rm {\bf R}}} \right|}{\left| {{\rm {\bf I}}+{\rm {\bf W}}_{2+} {\rm {\bf R}}} \right|}
\end{equation}
and ${\rm {\bf W}}_{i+} ={\rm {\bf P}}_+ {\rm {\bf W}}_i {\rm {\bf P}}_+ $, ${\rm {\bf P}}_+ ={\rm {\bf U}}_+ {\rm {\bf U}}_+^+ $ is the projection matrix on the positive eigenspace of ${\rm {\bf W}}_1 -{\rm {\bf W}}_2 $, ${\rm {\bf U}}_+ $ is a semi-unitary matrix whose columns are the eigenvectors of ${\rm {\bf W}}_1 -{\rm {\bf W}}_2 $ corresponding to its positive eigenvalues: ${\rm {\bf W}}_+ ={\rm {\bf U}}_+ {\rm {\bf D}}_+ {\rm {\bf U}}_+^+$, and ${\rm {\bf D}}_+ $ is the diagonal matrix of the positive eigenvalues; $C_+ ({\rm {\bf R}})$ is a non-negative, concave and non-decreasing function of ${\rm {\bf R}}$ or strictly positive, concave and increasing when the active eigenmodes of ${\rm {\bf R}}$ are in the span of the active eigenmodes of ${\rm {\bf W}}_+ $. The lower bound is tight (achieved with equality) when the channel is degraded or when $\bW_1$ and $\bW_2$ have the same eigenvectors, or in the low-SNR regime.
\end{prop}
\begin{proof}
see Appendix.
\end{proof}

The problem in \eqref{eq8a} has further significance: while the problem $C_s = \max_{{\rm {\bf R}}\ge 0} C({\rm {\bf R}})$ is not convex when the channel is not degraded, so that powerful tools of convex optimization \cite{Boyd} cannot be used, the problem $\max_{{\rm {\bf R}}\ge 0} C_+ ({\rm {\bf R}})$ is convex for any channel (degraded or not), to which all machinery of convex optimization can be applied and a lower bound (achievable rate) to the secrecy capacity can be evaluated using any standard convex solver.

\vspace*{0.5\baselineskip}
\section{Conclusion}
\label{sec:conclusion}
Optimal signalling over the Gaussian MIMO wire-tap channel has been studied under the total power constraint. A closed-form solution is given for the optimal transmit covariance matrix when the channel is strictly degraded. While the optimal signalling has some similarities to the conventional water-filling, it also reveals a number of differences: the optimal signalling does not converge to isotropic at high SNR. The weak eavesdropper and high-SNR regimes are considered, and a tighter upper bound on the rank of the optimal covariance matrix is given for the general case, along with the lower bound to the secrecy capacity, which is tight in a number of cases. While the general case is still an open problem (even when the channel is degraded), a characterization of an optimal covariance based on the active signaling subspace is given, which reveals hidden convexity in the underlying optimization problem.

\section{Appendix}
\label{sec:appendix}

\subsection{Proof of Theorem \ref{thm:Full-rank}}
Using the Lagrange multiplier technique
\cite{Boyd}\cite{Berstekas}, the optimization problem in \eqref{eq3} has the following Lagrangian:
\begin{equation}
\label{eqA1}
L=-\ln \left| {{\rm {\bf I}}+{\rm {\bf W}}_1 {\rm {\bf R}}} \right|+\ln
\left| {{\rm {\bf I}}+{\rm {\bf W}}_2 {\rm {\bf R}}} \right|+\lambda (tr{\rm {\bf R}}-P_T )-tr({\rm {\bf MR}})
\end{equation}
where $\lambda \ge 0 $ is a Lagrange multiplier responsible for the power
constraint $tr{\rm {\bf R}}\le P_T $ and ${\rm {\bf M}}\ge {\rm {\bf 0}}$ is a (positive semi-definite) matrix Lagrange multiplier responsible for the constraint ${\rm {\bf R}}\ge {\rm {\bf 0}}$. The associated KKT conditions (see e.g. \cite{Boyd}) can be expressed as:
\bal
\label{eqA2}
& \lambda ({\rm {\bf I}}+{\rm {\bf W}}_1 {\rm {\bf R}})({\rm {\bf I}}+{\rm {\bf RW}}_2 )={\rm {\bf W}}_1 -{\rm {\bf W}}_2 +{\rm {\bf M}} \\
\label{eqA3}
& {\rm {\bf MR = 0}}, \ \lambda (tr{\rm {\bf R}}-P_T )=0, \\
\label{eqA3a}
& {\rm {\bf R}} \ge {\rm {\bf 0}}, \ {\rm {\bf M}} \ge {\rm {\bf 0}}, \ \lambda \ge 0, \ tr{\rm {\bf R}}\le P_T
\eal
where (\ref{eqA2}) is obtained from $\partial L/\partial {\rm {\bf R}}={\rm {\bf 0}}$,
\begin{equation}
\label{eqA4}
\partial L/\partial {\rm {\bf R}}=({\rm {\bf I}}+{\rm {\bf W}}_2 {\rm {\bf R}})^{-1}{\rm {\bf W}}_2 -({\rm {\bf I}}+{\rm {\bf W}}_1 {\rm {\bf R}})^{-1}{\rm {\bf W}}_1 +\lambda {\rm {\bf I}}-{\rm {\bf M}}={\rm {\bf 0}}
\end{equation}
and the two equalities in (\ref{eqA3}) are the complementary slackness conditions while \eqref{eqA3a} are the primal and dual feasibility conditions.

Note that the (affine) constraints $tr{\rm {\bf R}}\le P_T $, ${\rm {\bf R}}\ge {\rm {\bf 0}}$ clearly satisfy the Slater condition
\cite{Boyd}\cite{Berstekas}. It also follows from Proposition \ref{prop:LB} that $C({\rm {\bf R}})$ is concave when ${\rm {\bf W}}_1 >{\rm {\bf W}}_2  $ (no need for projection) so that the problem in \eqref{eq3} is convex and thus the KKT conditions are sufficient for global optimality when the channel is strictly degraded.

Let us consider first the case of $\bW_2 >0$ and extend it to the singular case later. Assuming ${\rm {\bf R}}>{\rm {\bf 0}}$ and using ${\rm {\bf M}}={\rm {\bf 0}}$ (which follows from ${\rm {\bf MR}}={\rm {\bf 0}})$, one obtains from (\ref{eqA4}),
\begin{equation}
\label{eqD1}
{\rm {\bf R}}_1^{-1} -{\rm {\bf R}}_2^{-1} =\lambda {\rm {\bf I}}
\end{equation}
where ${\rm {\bf R}}_i ={\rm {\bf W}}_i^{-1} +{\rm {\bf R}}$, $i=1,2$. Let ${\rm {\bf R}}_1 ={\rm {\bf U\Lambda }}_1 {\rm {\bf U}}^+$ be the eigenvalue decomposition, where the columns of unitary matrix ${\rm {\bf U}}$ are the eigenvectors, and ${\rm {\bf \Lambda }}_1 >{\rm {\bf 0}}$ is a diagonal matrix of the corresponding eigenvalues. Using this in (\ref{eqD1}), one obtains ${\rm {\bf \Lambda }}_1^{-1} -({\rm {\bf U}}^+{\rm {\bf R}}_2  {\rm {\bf U}})^{-1}=\lambda {\rm {\bf I}}$ and therefore ${\rm {\bf U}}^+{\rm
{\bf R}}_2  {\rm {\bf U}}={\rm {\bf \Lambda }}_2  $ is diagonal, so that ${\rm {\bf R}}_2  ={\rm {\bf U\Lambda }}_2  {\rm {\bf U}}^+$ is the eigenvalue decomposition of ${\rm {\bf R}}_2  $, from which it follows that ${\rm {\bf R}}_1 $ and ${\rm {\bf R}}_2  $ have the same eigenvectors. Using this in (\ref{eqD1}) one obtains
\begin{equation}
\label{eqD2}
{\rm {\bf \Lambda }}_1 =(\lambda {\rm {\bf I}}+{\rm {\bf \Lambda }}_2^{-1} )^{-1}
\end{equation}
Furthermore,
\begin{equation}
\label{eqD3}
{\rm {\bf R}}_2  -{\rm {\bf R}}_1 ={\rm {\bf W}}_2^{-1} -{\rm {\bf
W}}_1^{-1} ={\rm {\bf U}}({\rm {\bf \Lambda }}_2  -{\rm {\bf \Lambda }}_1 ){\rm {\bf U}}^+
\end{equation}
so that the columns of ${\rm {\bf U}}$ are also the eigenvectors of ${\rm {\bf W}}_2^{-1} -{\rm {\bf W}}_1^{-1} = \bZ^{-1} $ and the diagonal entries of ${\rm {\bf \Lambda }}_2  -{\rm {\bf \Lambda }}_1 = diag\{\mu_i^{-1} \}$ are its eigenvalues. Combining the latter with (\ref{eqD2}), one obtains after some manipulations (\ref{eq10}). (\ref{eq9}) follows from ${\rm {\bf R}}_1 ={\rm {\bf W}}_1^{-1}
+{\rm {\bf R}}$ and ${\rm {\bf R}}_1 ={\rm {\bf U\Lambda }}_1 {\rm {\bf
U}}^+$. It is straightforward to see that $\lambda >0$ (otherwise ${\rm {\bf W}}_1 \le {\rm {\bf W}}_2  )$, so that transmission with the full power is optimal and (\ref{eq11}) follows from the power constraint $tr{\rm {\bf R}}=P_T $.
For (\ref{eq9}) to be a valid solution, we need ${\rm {\bf U\Lambda }}_1 {\rm {\bf U}}^+>{\rm {\bf W}}_1^{-1} $. This is insured by observing that the left-hand side of (\ref{eq11}) is monotonically decreasing in $\lambda $, so that the latter is monotonically decreasing as $P_T $ increases and, from (\ref{eq10}), $\lambda _{1i} $ also monotonically increases. Therefore, for sufficiently large $P_T$, $P_T > P_{T0}$ for some finite $P_{T0}$, the minimum eigenvalue of ${\rm {\bf \Lambda
}}_1 $ exceeds the maximum one of ${\rm {\bf W}}_1^{-1} $ and thus the
condition ${\rm {\bf U\Lambda }}_1 {\rm {\bf U}}^+>{\rm {\bf W}}_1^{-1} $ follows. Therefore, (\ref{eq9})-(\ref{eq11}) solve the KKT conditions and thus achieve the global optimum. It can be further seen that the solution is unique.

It can be seen that \eqref{eq11} is monotonically decreasing in $\lambda$ over the interval $(0,\infty)$ when $\lambda \in (0,\infty)$ so that a solution exists and unique for any $P_T$.

The condition ${\rm {\bf W}}_2  >{\rm {\bf 0}}$ can be further removed via the standard continuity argument \cite{Zhang}: use ${\bf W}_{2\epsilon} =  {\bf W}_2 + \epsilon {\bf I} > {\bf 0}$, $\epsilon > 0$, instead of ${\bf W}_{2}$ in Theorem \ref{thm:Full-rank} and then take $ \epsilon \rightarrow 0$. Alternatively, one may observe that ${\bf W}$ and ${\bf W}^{-1}$ have the same eigenvectors and inverse eigenvalues and use the matrix inversion lemma \cite{Zhang}\cite{Horn} to obtain:
\begin{equation}
({\bf W}_2^{-1} - {\bf W}_1^{-1})^{-1} = {\bf W}_2 + {\bf W}_2 ({\bf W}_1 - {\bf W}_2)^{-1} {\bf W}_2 = {\bf Z}
\end{equation}
Note that ${\bf Z}$ is well-defined even for singular ${\bf W}_2$ (since ${\bf W}_1 > {\bf W}_2$), its eigenvectors are those of ${\bf W}_2^{-1} - {\bf W}_1^{-1}$ and $\mu_i = \lambda_i(\bf Z)$ so that Theorem \ref{thm:Full-rank} applies. Furthermore, $\lambda_i(\bf Z) = 0$ iff $\lambda_i({\bf W}_2) = 0$, the corresponding eigenvectors are those of ${\bf W}_2$ and $\mu_i = 0$ implies $\lambda _{1i}=1/\lambda$. The equalities in  \eqref{eq12}  follow by observing that
\bal
|\bI+\bR^*\bW_1|=|\bW_1\bU\bLam_1\bU^+|=|\bW_1||\bLam_1|
\eal
and
\bal
|\bI+\bR^*\bW_2|=|\bI-\bW_2(\bW_1^{-1} - \bU\bLam_1\bU^+)|=|\bW_2||\bLam_2|
\eal
where 2nd equality holds when $\bW_2>0$ (1st one allows for singular $\bW_2$). Note that $\bW_2^{1/2}(\bW_1^{-1} - \bU\bLam_1\bU^+)\bW_2^{1/2} < \bI$ (which follows from $\bW_2^{1/2}\bW_1^{-1}\bW_2^{1/2} < \bI$ which in turn is implied by $\bW_1>\bW_2$) so that 2nd determinant is indeed strictly positive.

To show \eqref{eq.T2.5}, observe that $\bR^* > 0$. Using \eqref{eq9}, this requires $\bU \bLam_1 \bU^+ > \bW_1^{-1}$, which is insured by $\lambda_{1min} \lambda_{min} > 1$, where $\lambda_{1min} = \min_i \{\lambda_{1i} \}$ and $\lambda_{min}$ is the minimum eigenvalue of ${\bf W}_1$ (this follows from the fact that $\bW_1 > \bW_2$ is implied by $\lambda_{min}(\bW_1) > \lambda_{max}(\bW_2)$). Therefore, the threshold power $P_{T0}$ can be found from the boundary condition $\lambda_{1min}(P_{T0})=1/ \lambda_{min}$, which, after some manipulations, can be expressed as
\bal
\sqrt{\lambda^2+4 \mu_1\lambda}=2\lambda_{min}-\lambda
\eal
and can be solved for $\lambda$:
\bal
\lambda = \frac{\lambda_{min}^2}{\mu_1+\lambda_{min}}
\eal
Substituting this in \eqref{eq11}, one finally obtains \eqref{eq.T2.5}. \qed

\subsection{Proof of Corollary \ref{cor:weak MIMO-WTC}}
Using $\sqrt {1+x} \approx 1+x/2-x^2/8$ when $x\ll 1$ in (\ref{eq10}), one obtains $\lambda _{1i} \approx \lambda ^{-1}+\mu _i\lambda^{-2}$, and using this in (\ref{eq11}), one obtains $\lambda \approx m(P_T +tr{\rm {\bf W}}_1^{-1} )^{-1}$. The condition $x\ll 1$ is equivalent to
$\lambda/ \mu _i \gg 4$, which in turn is equivalent to (\ref{eq13}), and the latter also implies $\min _i \lambda _i ({\rm {\bf W}}_1 )\gg \max _i \lambda _i ({\rm {\bf W}}_2 )$ (i.e. the eavesdropper channel is indeed much weaker than the main one), from which it follows that ${\rm {\bf W}}_2^{-1} -{\rm {\bf W}}_1^{-1} \approx {\rm {\bf W}}_2^{-1} $, and applying these in (\ref{eq9}),
one obtains (\ref{eq14}). \qed

\subsection{Proof of Proposition \ref{prop:general case}}
Let $\bR^*$ be optimal covariance in \eqref{eq3}. Observe that
\bal
C_s &= C(\bR^*)\\
&= \ln \frac{|\bI+\bW_1 \bR^*|}{|\bI+\bW_2 \bR^*|}\\
\label{eq:prop.GC.3}
&= \ln \frac{|\bI+\bW_1 \bP_a\bR^*\bP_a|}{|\bI+\bW_2 \bP_a\bR^*\bP_a|}\\
\label{eq:prop.GC.4}
&= \ln \frac{|\bI+ \tilde{\bW}_1 \tilde{\bR}^*|}{|\bI+\tilde{\bW}_2 \tilde{\bR}^*|}\\
\label{eq:prop.GC.5}
&\le \max_{\tilde{\bR}} \ln \frac{|\bI+\tilde{\bW}_1 \tilde{\bR}|}{|\bI+\tilde{\bW}_2 \tilde{\bR}|}\ \mbox{s.t.} \ \tilde{\bR}\ge 0,\ tr\tilde{\bR} \le P_T
\eal
where $\bP_a = \bU_a\bU_a^+$ is the projection matrix on the subspace $span\{\bU_a\}$ and $\tilde{\bR}^*= \bU_a^+\bR^*\bU_a$; \eqref{eq:prop.GC.3} follows from $\bP_a\bR^*\bP_a = \bR^*$, \eqref{eq:prop.GC.4} follows from
\bal
|\bI+\bW_i \bP_a\bR^*\bP_a| = |\bI+ \bU_a^+ \bW_i \bU_a \bU_a^+\bR^*\bU_a|
\eal
\eqref{eq:prop.GC.5} follows from $tr\tilde{\bR}^* \le tr\bR^* \le P_T$ (since $\bU_a$ is semi-unitary). The 1st inequality in \eqref{eq:prop.GC.5} holds with equality, as can be proved by contradiction: assume that the inequality is strict so that
\bal
\ln \frac{|\bI+ \tilde{\bW}_1 \tilde{\bR}^*|}{|\bI+\tilde{\bW}_2 \tilde{\bR}^*|} = \ln \frac{|\bI+ \bW_1 \bU_a\tilde{\bR}^*\bU_a^+|}{|\bI+ \bW_2 \bU_a\tilde{\bR}^*\bU_a^+|} = C(\bR') > C(\bR^*)
\eal
where $\bR' = \bU_a\tilde{\bR}^*\bU_a^+$. Now note that $tr\bR' = tr\tilde{\bR}^* \le P_T$ so that $\bR'$ is feasible and hence the strict inequality is impossible. Further note that $\tilde{\bW}_1>\tilde{\bW}_2$ (this follows from \eqref{eq4a}) and that $\tilde{\bR}^*$ is of full rank. Therefore, the problems in \eqref{eq3} and \eqref{eq:prop.GC.5} are equivalent and Theorem \ref{thm:Full-rank} applies, from which the desired result follows. \qed

\subsection{Proof of Proposition \ref{prop:necessary cond}}
Observe that the KKT conditions in \eqref{eqA1}-\eqref{eqA3a} are not sufficient for optimality in the general (non-degraded) case since the original problem is not convex (see e.g. \cite{Boyd}). However, since the (affine) constraints $tr{\rm {\bf R}}\le P_T $, ${\rm {\bf R}}\ge {\rm {\bf
0}}$ clearly satisfy the Slater condition
\cite{Boyd}\cite{Berstekas} and since the maximum is achievable (since the constraint set is compact and the objective function is continuous), the KKT conditions are necessary for optimality \cite{Berstekas}. We further need the following technical Lemma.

\begin{lemma}
\label{lemma:ABC}
Let ${\rm {\bf A}},{\rm {\bf B}},{\rm {\bf C}}\ge {\rm {\bf 0}}$ be positive semi-definite matrices and let ${\rm {\bf ABC}}$ be Hermitian. Then ${\rm {\bf ABC}}\ge {\rm {\bf 0}}$.
\end{lemma}
\begin{proof}
Since ${\rm {\bf A}},{\rm {\bf C}}\ge {\rm {\bf 0}}$, there exists a non-singular matrix ${\rm {\bf S}}$ such that ${\rm {\bf
SAS}}^+={\rm {\bf D}}_a \ge {\rm {\bf 0}},{\rm {\bf SCS}}^+={\rm {\bf D}}_c \ge {\rm {\bf 0}}$ are diagonal \cite{Zhang}. Using the
latter,
\bal
{\rm {\bf ABC}}={\rm {\bf SD}}_a \overline {\rm {\bf B}} {\rm {\bf D}}_c {\rm {\bf S}}^+
\eal
where $\overline {\rm {\bf B}} ={\rm {\bf S}}^+{\rm {\bf BS}}\ge {\rm {\bf 0}}$. Observe further that
\bal
\lambda _i ({\rm {\bf D}}_a \overline {\rm {\bf B}} {\rm {\bf D}}_c ) &= \lambda _i (\overline {\rm {\bf B}} {\rm {\bf D}}_c {\rm {\bf D}}_a )\\
&=\lambda _i (({\rm {\bf D}}_c {\rm {\bf D}}_a )^{1/2}\overline {\rm {\bf B}} ({\rm {\bf D}}_c {\rm {\bf D}}_a )^{1/2})\ge \bf{0}
\eal
since $({\rm {\bf D}}_c {\rm {\bf D}}_a)^{1/2}\overline {\rm {\bf B}} ({\rm {\bf D}}_c {\rm {\bf D}}_a )^{1/2}\ge {\rm {\bf 0}}$, where $\lambda _i ({\rm {\bf B}})$ means an eigenvalue of
matrix ${\rm {\bf B}}$. Since ${\rm {\bf D}}_a \overline {\rm {\bf B}} {\rm {\bf D}}_c $ is Hermitian (because ${\rm {\bf ABC}}$ is) and has
non-negative eigenvalues, it is positive semi-definite \cite{Zhang}, ${\rm {\bf D}}_a \overline {\rm {\bf B}} {\rm {\bf D}}_c \ge {\rm {\bf 0}}$. It follows that ${\rm {\bf ABC}}={\rm {\bf SD}}_a \overline {\rm {\bf B}} {\rm {\bf D}}_c {\rm {\bf S}}^+\ge {\rm {\bf 0}}$.
\end{proof}

Note that this Lemma is a generalization of a well known fact: ${\rm {\bf AB}}\ge {\rm {\bf 0}}$ if ${\rm {\bf A}},{\rm {\bf B}}\ge {\rm {\bf 0}}$ and ${\rm {\bf AB}}$ is Hermitian \cite{Zhang}. We first prove that ${\rm {\bf Z}}=({\rm {\bf I}}+{\rm {\bf W}}_1 {\rm {\bf R}})({\rm {\bf I}}+{\rm {\bf RW}}_2 )>{\rm {\bf 0}}$ when ${\rm {\bf R}}>{\rm {\bf 0}}$. In this case, ${\bf Z}$ can be expressed as
\begin{equation}
\label{eqA5}
{\rm {\bf Z}}=({\rm {\bf R}}^{-1}+{\rm {\bf W}}_1 ){\rm {\bf R}}^2({\rm {\bf
R}}^{-1}+{\rm {\bf W}}_2 )
\end{equation}
Now identify the right-hand side of (\ref{eqA5}) with ${\rm {\bf A}},{\rm {\bf B}},{\rm {\bf C}}$ and use Lemma \ref{lemma:ABC} to obtain ${\rm {\bf Z}}\ge {\rm {\bf 0}}$ (noting that $\bf Z$ is Hermitian from \eqref{eqA2}).
Therefore, it follows from (\ref{eqA2}) that
\bal
{\rm {\bf W}}_1 -{\rm {\bf W}}_2 +{\rm {\bf M}}\ge {\rm {\bf 0}}
\eal
since $\lambda >0$, as $\lambda =0$ implies ${\rm {\bf W}}_1 \le {\rm {\bf W}}_2 $ and thus $C_s =0$ - trivial case not considered here. Since $\left| {({\rm {\bf I}}+{\rm {\bf W}}_1 {\rm {\bf R}})({\rm {\bf I}}+{\rm {\bf RW}}_2 )} \right|>0$, it further follows that ${\rm {\bf Z}}>{\rm {\bf 0}}$ and
\bal
{\rm {\bf W}}_1 -{\rm {\bf W}}_2 +{\rm {\bf M}}>{\rm {\bf 0}}.
\eal

The case of singular $\bf R$ is somewhat more involved. Let ${\bf R}={\bf U\Lambda U^+}$ be the eigenvalue decomposition of ${\bf R}$. Consider
\begin{equation}
\label{eqA6}
{\bf \tilde {Z}}={\bf U}^+{\bf ZU}=({\bf I}+{\bf \tilde {W}}_1 {\bf \Lambda })({\bf I}+{\bf \Lambda \tilde {W}}_2 )={\bf \tilde {W}}_1 -{\bf \tilde {W}}_2 +{\bf \Lambda }_M
\end{equation}
where ${\bf \tilde {W}}_i ={\bf U}^+{\bf W}_i {\bf U}$, ${\bf \Lambda }_M ={\bf U}^+{\bf MU}$, and block-partition ${\bf \Lambda },{\bf \tilde {W}}_i $ as follows:
\begin{equation}
\label{eqA7}
{\bf \Lambda }=\left[ {{\begin{array}{*{20}c}
 {{\bf \Lambda }_r } \hfill & {\bf 0} \hfill \\
 {\bf 0} \hfill & {\bf 0} \hfill \\
\end{array} }} \right],{\bf \tilde {W}}_i =\left[ {{\begin{array}{*{20}c}
 {{\bf W}_i^{11} } \hfill & {{\bf W}_i^{12} } \hfill \\
 {{\bf W}_i^{21} } \hfill & {{\bf W}_i^{22} } \hfill \\
\end{array} }} \right]
\end{equation}
where ${\bf \Lambda }_r $ is a diagonal matrix collecting $r$ positive eigenvalues of ${\bf R}$. Using this in \eqref{eqA6}, one obtains, after some manipulations,
\begin{equation}
\label{eqA8}
{\bf \tilde {Z}}=\left[ {{\begin{array}{*{20}c}
 {({\bf W}_1^{11} {\bf \Lambda }_r +{\bf I}_r )({\bf \Lambda }_r {\bf W}_2^{11} +{\bf I}_r )} \hfill & {({\bf W}_1^{11} {\bf \Lambda }_r +{\bf I}_r ){\bf \Lambda }_r {\bf W}_2^{12} } \hfill \\
 {{\bf W}_1^{21} {\bf \Lambda }_r ({\bf \Lambda }_r {\bf W}_2^{11} +{\bf I}_r )} \hfill & {{\bf W}_1^{21} {\bf \Lambda }_r^2 {\bf W}_2^{12} +{\bf I}_r } \hfill \\
\end{array} }} \right]
\end{equation}
where ${\bf I}_r $ is $ r \times r$ identity matrix. Note that ${\bf \tilde {Z}}$ is Hermitian (since ${\bf Z}$ is) and use the following fact \cite{Zhang}:
\begin{equation}
\label{eqA9}
\left[ {{\begin{array}{*{20}c}
 {\bf A} \hfill & {\bf B} \hfill \\
 {{\bf B}^+} \hfill & {\bf X} \hfill \\
\end{array} }} \right]\ge {\bf 0}\leftrightarrow {\bf X}\ge {\bf B}^+{\bf A}^{-1}{\bf B}
\end{equation}
where ${\bf X},{\bf A}$ are Hermitian (and so is the block-partitioned matrix) and $\leftrightarrow $ means that the conditions are equivalent. Apply this to \eqref{eqA8} to obtain
\begin{align}
\label{eqA10} \notag
{\bf B}^+{\bf A}^{-1}{\bf B} &= {\bf W}_1^{21} {\bf \Lambda }_r ({\bf \Lambda }_r {\bf W}_2^{11} +{\bf I}_r )(({\bf W}_1^{11} {\bf \Lambda }_r +{\bf I}_r )({\bf \Lambda }_r {\bf W}_2^{11} +{\bf I}_r ))^{-1}({\bf W}_1^{11} {\bf \Lambda }_r +{\bf I}_r ){\bf \Lambda }_r {\bf W}_2^{12} \\
&= {\bf W}_1^{21} {\bf \Lambda }_r^2 {\bf W}_2^{12} \le {\bf W}_1^{21} {\bf \Lambda }_r^2 {\bf W}_2^{12} +{\bf I}_r ={\bf X}
\end{align}
so that ${\bf \tilde {Z}}\ge {\bf 0}$ and thus ${\bf Z}\ge {\bf 0}$ follow. Since $|{\bf Z}|\ne {\bf 0}$, it further follows that ${\bf Z}>{\bf 0}$ and thus
\begin{align}
{\bf W}_1 -{\bf W}_2 +{\bf M}>{\bf 0}
\end{align}

To prove \eqref{eq4a}, note that
\begin{align}
{\bf 0} < {\bf U}_{r+}^+ ({\bf W}_1 - {\bf W}_2 + {\bf M} ){\bf U}_{r+} = {\bf U}_{r+}^+ ({\bf W}_1 - {\bf W}_2){\bf U}
\end{align}
where the columns of ${\bf U}_{r+}$ are the active eigenvectors $\{{\bf u}_{i+}\}$. The inequality follows since ${\bf W}_1 - {\bf W}_2 + {\bf M} > {\bf 0}$ and the columns of ${\bf U}_{r+}$ being linearly independent:
\begin{align}
{\bf x}^+ {\bf U}_{r+}^+ ({\bf W}_1 - {\bf W}_2 + {\bf M} ){\bf U}_{r+} {\bf x} = {\bf \tilde{x}}^+ ({\bf W}_1 - {\bf W}_2 + {\bf M}){\bf \tilde{x}} > {\bf 0} \ \forall {\bf x \neq 0}
\end{align}
where ${\bf \tilde{x}} = {\bf U}_{r+} {\bf x} \neq {\bf 0}$ since the columns of ${\bf U}_{r+}$ are linearly independent. The equality follows since ${\bf MR}={\bf 0}$ implies ${\bf M}{\bf U}_{r+} = \bf 0$.  \eqref{eq4} follows from \eqref{eq4a} by expressing $\bx=\bU_{r+}\bz$ for some $\bz$.

\subsection{Proof of Proposition \ref{prop:LB}}
We will need the following technical Lemma.

\textbf{Lemma 3}: Consider the function
\[
f({\rm {\bf X}})=\ln \left| {{\rm
{\bf I}}-{\rm {\bf B}}({\rm {\bf A}}+{\rm {\bf X}})^{-1}{\rm {\bf B}}}
\right|,
\]
where ${\rm {\bf A}},{\rm {\bf B}},{\rm {\bf X}}\ge {\rm {\bf 0}}$ are positive semi-definite matrices, ${\rm {\bf I}}$ is the identity matrix, ${\rm {\bf BA}}^{-1}{\rm {\bf B}}\le {\rm {\bf I}}$. It has the following properties:

\begin{enumerate}
\item $f({\rm {\bf X}})$ is increasing in ${\rm {\bf X}}$: ${\rm {\bf X}}_1 \le {\rm {\bf X}}_2 \to f({\rm {\bf X}}_1 )\le f({\rm {\bf X}}_2 )$.
\item $f({\rm {\bf X}})$ is concave in ${\rm {\bf X}}$:
 \[
 f(\alpha {\rm {\bf X}}_1 +\beta {\rm {\bf X}}_2 )\ge \alpha f({\rm {\bf X}}_1 )+\beta f({\rm {\bf X}}_2 ),
  \]
  for $\alpha +\beta =1, \ 0\le \alpha , \ \beta \le 1$.
 \end{enumerate}

\textbf{Proof: }1st property follows from the (easy to verify) fact that $-{\rm {\bf B}}({\rm {\bf A}}+{\rm {\bf X}})^{-1}{\rm {\bf B}}$ is
increasing in ${\rm {\bf X}}$ (in the matrix positive definite ordering sense
\cite{Zhang}). 2nd one is obtained from the following chain argument:
\begin{eqnarray}
\label{eqC1}
f(\alpha {\rm {\bf X}}_1 + \beta {\rm {\bf X}}_2 ) &=& \ln \left| {{\rm {\bf I}}-{\rm {\bf B}}({\rm {\bf A}}+\alpha {\rm {\bf X}}_1 +\beta {\rm {\bf X}}_2 )^{-1}{\rm {\bf B}}} \right| \mbox{} \\ \notag
&{\mathop \ge \limits^{(a)}}&
\ln \left| {{\rm {\bf I}}-\alpha {\rm {\bf B}}
{\rm {\bf A}}_1^{-1}{\rm {\bf B}}-\beta {\rm {\bf
B}}{\rm {\bf A}}_2^{-1}{\rm {\bf B}}} \right| \\ \notag
&{\mathop \ge \limits^{(b)}}& \alpha \ln \left| {{\rm {\bf I}}-{\rm {\bf B}} {\rm {\bf A}}_1^{-1}{\rm {\bf B}}} \right| + \beta \ln \left| {{\rm {\bf I}}-{\rm {\bf B}} {\rm {\bf A}}_2^{-1}{\rm {\bf B}}} \right| \\ \notag
&=& \alpha f({\rm {\bf X}}_1 )+\beta f({\rm {\bf X}}_2 )
\end{eqnarray}
where ${\rm {\bf A}}_i = {\rm {\bf A}}+{\rm {\bf X}}_i$; (a) follows from the facts that $F({\rm {\bf X}})={\rm {\bf X}}^{-1}$ is convex in ${\rm {\bf X}}$ and $F({\rm {\bf X}})=\ln \left| {\rm {\bf X}} \right|$ is increasing \cite{Boyd}\cite{Zhang}; (b) follows from the fact that $F({\rm {\bf X}})=\ln \left| {\rm {\bf X}} \right|$ is concave \cite{Boyd}. \qed

We now assume that ${\rm {\bf W}}_{i+} >{\rm {\bf 0}}$. The case of singular ${\rm {\bf W}}_{i+} $ will follow from the standard continuity argument \cite{Zhang} (i.e. use ${\bf W}_{i\epsilon} = {\bf W}_{i+} + \epsilon {\bf I}$, $\epsilon > 0$, instead of ${\bf W}_{i+}$ and then take $\epsilon \rightarrow 0$; see section 2.6 in \cite{Zhang} for more details and examples). Observe that
\begin{eqnarray}
\label{eqC2}
 C_+ ({\rm {\bf R}}) &=& \ln \frac{\left| {{\rm {\bf W}}_{1+} }
\right|}{\left| {{\rm {\bf W}}_{2+} } \right|} + \ln \frac{\left| {{\rm {\bf W}}_{1+}^{-1} +{\rm {\bf R}}} \right|}{\left| {{\rm {\bf W}}_{2+}^{-1} +{\rm {\bf R}}} \right|} \\ \notag
&=& c + \ln \left| {{\rm {\bf I}}- {\rm \bf{ \Delta W}} ({\rm {\bf W}}_{2+}^{-1} +{\rm {\bf R}})^{-1}} \right| \\ \notag
&=& c + \ln \left| {{\rm {\bf I}}- {\rm \bf{ \Delta W}}^{1/2}({\rm {\bf W}}_{2+}^{-1} +{\rm {\bf R}})^{-1} {\rm \bf{ \Delta W}}^{1/2}} \right|
\end{eqnarray}
where $c = \ln \left| {{\rm {\bf W}}_{1+} }\right| - \ln \left| {{\rm {\bf W}}_{2+} } \right|$ and $ {\rm \bf{ \Delta W}} = {\rm {\bf W}}_{2+}^{-1} -{\rm {\bf W}}_{1+}^{-1}$,
and apply Lemma 3 to the last term of the last expression in \eqref{eqC2}. It is easy to verify that ${\rm {\bf BA}}^{-1}{\rm {\bf B}}\le {\rm {\bf I}}$ (since ${\rm {\bf W}}_{2+}^{-1} -{\rm {\bf W}}_{1+}^{-1} \le {\rm {\bf W}}_{2+}^{-1} )$ and that ${\rm {\bf B}}\ge {\rm {\bf 0}}$ (since ${\rm {\bf W}}_{1+} \ge
{\rm {\bf W}}_{2+})$, so that the properties of $C_+(\bR)$ follow.
To prove the lower bound, note that the problem in \eqref{eq8a}   limits the optimization to the positive eigenspace of $\bW_1-\bW_2$ and thus is sub-optimal. To prove the achievability of the lower bound, note that, in the low-SNR regime, one obtains  $C(\bR) \approx tr(\bW_1-\bW_2)\bR$ so that rank-1 transmission on the largest eigenmode of $\bW_1-\bW_2$ is optimal. But this eigenmode is in the positive eigenspace of $\bW_1-\bW_2$ (unless it is negative, in which case the capacity is zero) so that this transmission is also optimal for the projected problem. When eigenvectors of $\bW_1$ and $\bW_2$ are the same, the achievability follows from the respective result for parallel channels in \cite{Khisti-08}\cite{Li'10} (since an optimal covariance also has the same eigenvectors). When the channel is degraded, the projection has no effect since $\bW_1-\bW_2 \ge 0$ so that the problems in \eqref{eq3} and \eqref{eq8a} are identical. \qed


\end{document}